\documentclass[12pt,a4paper]{article}

\usepackage{mathptmx}       % selects Times Roman as basic font
\usepackage{helvet}         % selects Helvetica as sans-serif font
\usepackage{courier}        % selects Courier as typewriter font
\usepackage{type1cm}        % activate if the above 3 fonts are
\usepackage{makeidx}         % allows index generation
\usepackage{graphicx}        % standard LaTeX graphics tool
\usepackage{multicol}        % used for the two-column index
\usepackage[bottom]{footmisc}% places footnotes at page bottom
\usepackage{amsmath,amssymb}
\usepackage{float}
\usepackage{subfig}
\usepackage{amsthm}
\newtheorem{thm}{Theorem}[section]
\usepackage[textheight=23cm,textwidth=16cm]{geometry}

\makeindex             % used for the subject index
\begin{document}

\title{Technical report\\The computation of Greeks with \\Multilevel Monte Carlo}
\date{December 2010}
\author{Sylvestre Burgos, M.B.~Giles\\Oxford-Man Institute of Quantitative Finance\\University of Oxford}

%\institute{Sylvestre Burgos \at Oxford-Man Institute of Quantitative Finance, University of Oxford,\\ \and M.B.~Giles \at Oxford-Man Institute of Quantitative Finance, University of Oxford, \\}

\maketitle

%Book's comment
\abstract{We study the use of the multilevel Monte Carlo technique \cite{burgos:giles07b,burgos:giles08b} \linebreak in the context of the calculation of Greeks.
The pathwise sensitivity analysis \cite{burgos:glasserman04} differentiates the path evolution and reduces the  payoff's smoothness.
This leads to new challenges: 
the inapplicability of pathwise sensitivities to non-Lipschitz payoffs often makes the use of naive algorithms  impossible.\\
These challenges can be addressed in three different ways:
payoff smoothing using conditional expectations of the payoff before maturity \cite{burgos:glasserman04};
approximating the previous technique with path splitting for the final timestep \cite{burgos:ag07};
using of a hybrid combination of pathwise sensitivity and the Likelihood Ratio Method \cite{burgos:giles09}.
We investigate the strengths and weaknesses of these alternatives in different multilevel Monte Carlo settings.}

\section{Introduction}
\label{burgos:1}

In mathematical finance, Monte Carlo methods are used to compute the price of an option by estimating the expected value $\mathbb{E}(P)$. $P$ is the payoff function that depends on an underlying asset's scalar price $S(t)$ which satisfies an evolution SDE of the form
\begin{equation}
\mathrm{d} S(t)= a(S,t)\, \mathrm{d} t+ b(S,t)\, \mathrm{d} W_t, 
\quad 0 \leq t \leq T, 
\quad S(0) \textrm{ given}, \label{burgos:underlyingSDE}
\end{equation}
This is just one use of Monte Carlo in finance. In practice the prices are often quoted and used to calibrate our market models; the option's sensitivities to market parameters, the so-called Greeks, reflect the exposure to different sources of risk. Computing these is essential to hedge portfolios and is therefore even more important than pricing the option itself. This is why our research focuses on getting fast and accurate estimates of Greeks through Monte Carlo simulations.

\subsection{Multilevel Monte Carlo}
\label{burgos:MultilevelMonteCarlo}
%We consider equation \eqref{burgos:underlyingSDE}. The multilevel Monte Carlo technique works in multidimensional settings, nevertheless the use of high-order discretisation schemes then requires some special treatment and is currently being invetigated by Giles. For the sake of simplicity, we restrict our study to the case where $S $ is scalar.

Let us consider a standard Monte Carlo method using a discretisation with first order weak convergence (e.g. the Milstein scheme). Achieving a root-mean square error of $O(\epsilon)$ requires a variance of order $O(\epsilon^2)$, hence $O(\epsilon^{-2})$ independent paths. It also requires a discretisation bias of order $O(\epsilon^{-1})$, thus $O(\epsilon^{-1})$ timesteps, giving a total computational cost $O(\epsilon^{-3})$.

Giles' multilevel Monte Carlo technique reduces this cost to $O(\epsilon^{-2})$ under certain conditions. 
The idea is to write the expected payoff with a fine discretisation using $2^{-L}$ uniform timesteps as a telescopic sum.
Let $\hat P_l$ be the simulated payoff with a discretisation using $2^l$ uniform timesteps,
\begin{equation}
\mathbb{E}(\hat P_L)=\mathbb{E}(\hat P_0)+
\sum \limits_{l=1}^L \mathbb{E}(\hat P_l-\hat P_{l-1})
\label{burgos:telescope}
\end{equation}
We then use Monte Carlo estimators using $N_l$ independent samples
\begin{equation}
\mathbb{E}(\hat P_l-\hat P_{l-1}) \approx \hat Y_l =
\frac{1}{N_l} \sum \limits_{i=1}^{N_l}
(\hat P_l^{(i)}-\hat P_{l-1}^{(i)})
\label{burgos:MultilevelEstimator}
\end{equation}
The small corrective term $\hat P_l^{(i)}-\hat P_{l-1}^{(i)}$ comes from the difference between a fine and a coarse discretisation of the same leading Brownian motion. Its magnitude depends on the strong convergence properties of the scheme used. Let $V_l$ be the variance of a single sample $\hat P_l^{(i)}-\hat P_{l-1}^{(i)}$. The next theorem shows that what determines the efficiency of the multilevel approach is the convergence speed of $V_l$ as $l \rightarrow \infty$.

To ensure a better efficiency we may modify \eqref{burgos:MultilevelEstimator} and use different estimators of $\hat P$ on the fine and coarse levels  of $\hat Y_l$ as long as the telescoping sum property is respected, that is
\begin{equation}
\mathbb{E}(\hat P_L)=\mathbb{E}(\hat P_0)+
\sum \limits_{l=1}^L \mathbb{E}(\hat P_l^{f}-\hat P_{l-1}^{c}) \label{burgos:NewMultilevelEstimator}
\end{equation}
with
\begin{equation}
\mathbb{E}(\hat P_l^{f}) = \mathbb{E}(\hat P_{l}^{c}).
\label{burgos:TelescopeCondition}
\end{equation}

\begin{thm}
\label{burgos:complexity_thm}
Let $P$ be a function of a solution to \eqref{burgos:underlyingSDE} for a given Brownian path $W(t)$; let $\hat P_l$ be the corresponding approximation using the discretisation at level $l$, i.e. with $2^l$ steps of width $h_l=2^{-l} \, T$.

If there exist independent estimators $\hat Y_l$ of computational complexity $C_l$ based on $N_l$ samples and there are positive constants $ \alpha~\geq~\frac{1}{2}, \beta, c_1, c_2, c_3$ such that

\begin{enumerate}
\item $\displaystyle \mathbb{E}(\hat Y_l)=\left\{ \begin{array}{ll}
 \mathbb{E}(\hat P_0) &\mbox{ if }l=0 \\
 \mathbb{E}(\hat P_l- \hat P_{l-1}) &\mbox{ if }l >0\\
 \end{array} \right.$
\item $\displaystyle |\mathbb{E}(\hat P_l - P)| \leq c_1 h_l^{\alpha}$
\item $\displaystyle \mathbb V (\hat Y_l) \leq c_2 h_l^{\beta}\,N_l^{-1}$
\item $\displaystyle C_l \leq c_3 N_l \,h_l^{-1}$
\end{enumerate}

Then there is a constant $c_4$ such that for any $\epsilon < e^{-1}$, there are values for $L$ and $N_l$ resulting in a multilevel estimator $\hat Y = \displaystyle \sum_{l=0}^{L} \hat Y_l$ with a mean-square-error $MSE~=~\mathbb E((\hat Y -~\mathbb E(P))^2) < \epsilon ^2$ with a complexity $C$ bounded by

\begin{equation}
\label{burgos:convergence_speed}
C \leq 
\left\{ \begin{array}{ll}
\displaystyle c_4 \epsilon^{-2}&\mbox{ if }\beta >1 \\
\displaystyle c_4 \epsilon^{-2} \left( \log \epsilon \right)^2&\mbox{ if }\beta=1\\
\displaystyle c_4 \epsilon^{-2-(1-\beta)/\alpha}&\mbox{ if }0< \beta < 1\\
\end{array} \right.
\end{equation}
\end{thm}
\begin{proof}
See \cite{burgos:giles08b}.
\end{proof}

We usually know $\alpha$ thanks to the literature on weak convergence. Results in \cite{burgos:kp92} give $\alpha=1$ for the Milstein scheme, even in the case of discontinuous payoffs. $\beta$ is related to strong convergence and is practically what determines the efficiency of the multilevel approach. Its value depends on the payoff shape and is usually not known \textit{a priori}.
% A rule of thumb is that discontinuous payoffs will lead to lower values of $\beta$ and will require a special treatment to regain optimal benefits from the multilevel technique.

\subsection{Monte Carlo Greeks}
\label{burgos:MonteCarloGreeks}

Let us briefly recall two classic methods used to compute Greeks in a Monte Carlo setting: the pathwise sensitivities and the Likelihood Ratio Method. More details can be found in \cite{burgos:glasserman04}.

\subsubsection*{Pathwise sensitivities}
Let $\displaystyle \hat{S}=(\hat{S}_n)_{n \in [0,N]}$ be the simulated values of the asset at the discretisation times and $\displaystyle \hat W=~(\hat W_n)_{n \in [1,N]}$ be the corresponding set of Brownian increments. Assuming that the payoff $P(\hat S)$ is Lipschitz, we can use the chain rule and write
%\begin{eqnarray}
%\displaystyle \frac{\partial \hat V}{\partial \theta} 
%&= \displaystyle \frac{\partial}{\partial \theta}  \displaystyle \int{P(\hat S(\theta, \hat W))\,  p(\hat W) \mathrm{d} \hat W}
%%= \int{\frac{\partial P(\hat S(\theta, \hat W)}{\partial \theta} \, p(\hat W) \mathrm{d} \hat W}
%= \displaystyle \int{\frac{\partial P(\hat S)}{\partial \hat S} \frac{ \partial \hat S(\theta, \hat W)}{\partial \theta} \, p(\hat W) \mathrm{d} \hat W}
%\\ \nonumber
% & \textrm{with}\qquad \displaystyle \mathrm{d}\hat{W} = \prod\limits_{k=1}^N {\mathrm{d}\hat{W}_k}\qquad \textrm{  and  }\qquad \displaystyle p(\hat{W}) = \prod\limits_{k=1}^N {p(\hat{W}_{k})}
%\end{eqnarray}
\[
\frac{\partial \hat V}{\partial \theta} 
= \frac{\partial}{\partial \theta}  \displaystyle \int{P(\hat S(\theta, \hat W))\,  p(\hat W) \mathrm{d} \hat W}
= \int{\frac{\partial P(\hat S)}{\partial \hat S} \frac{ \partial \hat S(\theta, \hat W)}{\partial \theta} \, p(\hat W) \mathrm{d} \hat W}
\]
where $\mathrm{d}\hat{W} = \prod\limits_{k=1}^N {\mathrm{d}\hat{W}_k}$
and   $p(\hat{W}) = \prod\limits_{k=1}^N {p(\hat{W}_{k})}$ is a 
multivariate Normal probability density function.

We obtain $\displaystyle  \frac{\partial \hat S}{\partial \theta}$ by differentiating the discretisation of \eqref{burgos:underlyingSDE} with respect to $\theta$ and iterating the resulting formula.The limitation of this technique is that it requires the payoff to be Lipschitz and piecewise differentiable.

\subsubsection*{Likelihood Ratio Method}
\label{burgos:lrm}
The Likelihood Ratio Method consists in writing
\begin{equation} 
V=\mathbb{E}\left[ P(\hat S) \right] = \int P(\hat S) p(\theta,\hat S) \, \mathrm{d} \hat S 
\end{equation}
The dependence on $\theta$ comes through the probability density function $p(\theta,\hat S)$; assuming some conditions discussed in \cite{burgos:glasserman04}, we can write
\begin{eqnarray}
&\displaystyle \frac{\partial V}{\partial \theta} =  \int P(\hat{S}) \frac{\partial p(\hat{S})}{\partial \theta} \,\mathrm{d} \hat{S} =\int P(\hat{S}) \frac{\partial \log p(\hat{S})}{\partial \theta} p(\hat{S}) \,\mathrm{d}\hat{S} = \mathbb{E}\left[ P(\hat{S}) \frac{\partial \log p(\hat{S})}{\partial \theta} \right]
\\ \nonumber
& \textrm{with}\qquad \displaystyle  \mathrm{d}\hat{S} = \prod\limits_{k=1}^N {\mathrm{d}\hat{S}_k} \qquad \textrm{  and  }\qquad \displaystyle  p(\hat{S}) = \prod\limits_{k=1}^N {p(\hat{S}_{k}|\hat{S}_{k-1})}
\end{eqnarray}

The main limitation of the method is that the estimator's variance is $O(N)$, becoming infinite as we refine the discretisation.

\subsection{Multilevel Monte Carlo Greeks}

By combining the elements of sections \ref{burgos:MultilevelMonteCarlo} and \ref{burgos:MonteCarloGreeks} together, we write
\begin{equation}
\displaystyle \frac{\partial V}{\partial \theta}=\frac{\partial \mathbb{E}( P)}{\partial \theta} \approx \frac{\partial \mathbb{E}(\hat P_L)}{\partial \theta}=\frac{\partial \mathbb{E}(\hat P_0)}{\partial \theta}+\displaystyle\sum_{l=1}^{L} \frac{\partial \mathbb{E}(\hat P_l -\hat P_{l-1})}{\partial \theta}
\end{equation}
As in \eqref{burgos:MultilevelEstimator}, we define the multilevel estimators
\begin{equation}
%\begin{cases}
\hat Y_0=N_0^{-1} \, \displaystyle \sum_{i=1}^{M} {\frac{\partial \hat P_0}{\partial \theta} }^{(i)}  \qquad
\textrm{and} \qquad
\hat Y_l=N_l^{-1} \, \displaystyle \sum_{i=1}^{N_l} ({\frac{\partial \hat P_l}{\partial \theta} }^{(i)} -{\frac{\partial \hat P_{l-1}}{\partial \theta} }^{(i)})
%,\, 1 \leq l \leq L
%\end{cases}
\end{equation}
where $\displaystyle  {\frac{\partial \hat P_0}{\partial \theta} }$, $\displaystyle  {\frac{\partial \hat P_{l-1}}{\partial \theta} }$, $\displaystyle  {\frac{\partial \hat P_l}{\partial \theta} }$ are computed with the techniques presented in section \ref{burgos:MonteCarloGreeks}.

\section{European call}
\label{burgos:2}
We first consider a Lipschitz payoff. That of the European call is 
\begin{equation}\nonumber P=(S_T-K)^+=\max(0,S_T-K)\end{equation}
 We illustrate the techniques by computing delta ($\delta$) and vega ($\nu$), the sensitivities to the asset's initial value $S_0$ and to its volatility $\sigma$.

\subsection{Pathwise sensitivities}
\label{burgos:pws21}
We consider the Black-Scholes model: the asset's evolution is modelled by a geometric Brownian motion
$ \mathrm{d} S(t)= r \, S(t) \mathrm{d} t+ \sigma \, S(t) \mathrm{d} W_t $.
We use the Milstein scheme for its good strong convergence properties. For  timesteps of width $h$,
\begin{equation}
\hat S_{n+1}=\hat S_n \cdot \left(1+r\,h+\sigma \, \Delta W_n + \frac{\sigma^2 }{2}\, (\Delta W_n^2-h)  \right) := \hat S_n \cdot D_n  \label{burgos:milstein}
\end{equation}
Since the payoff is Lipschitz, we can use pathwise sensitivities. The differentiation of equation \eqref{burgos:milstein} gives
\begin{equation}
\begin{cases}
\displaystyle  \frac{\partial \hat S_0}{\partial S_0}=1 , \quad \frac{\partial \hat S_{n+1}}{\partial S_0}=\frac{\partial \hat S_{n}}{\partial S_0} \cdot D_n\\
\displaystyle  \frac{\partial \hat S_0}{\partial \sigma}=0 , \quad \frac{\partial \hat  S_{n+1}}{\partial \sigma}=\frac{\partial \hat S_{n}}{\partial \sigma}\cdot D_n+\hat S_{n}\left(\Delta W_n+\sigma (\Delta W_n^2 - h)\right)
\end{cases}
\end{equation}
To compute $\hat Y_l$ we use a fine and a coarse discretisation with $N_f=2^l$ and $N_c=2^{l-1}$ uniform timesteps respectively. We denote these by the superscripts $^{(l)}$ and $^{(l-1)}$. We take $N_l$ samples to compute
%\begin{equation}
%\hat Y_l = \frac{1}{N_l}\sum_{i=1}^{N_l} \left[ \left(\frac{\partial \hat S_{N_f}}{\partial \theta}\frac{\partial P}{\partial S_{N_f}} \right)^{(l)} - \left( \frac{\partial \hat  S_{N_c}}{\partial \theta}\frac{\partial P}{\partial S_{N_c}} \right)^{(l-1)} \right]
%\end{equation}
\[
\hat Y_l = \frac{1}{N_l}\sum_{i=1}^{N_l} \left[ \left(
\frac{\partial P}{\partial S_{N_f}} 
\frac{\partial \hat S_{N_f}}{\partial \theta}^{(i)}
\right)^{(l)} - \left( 
\frac{\partial P}{\partial S_{N_c}} 
\frac{\partial \hat  S_{N_c}}{\partial \theta}^{(i)}
\right)^{(l-1)} \right]
\]
We use the same leading Brownian motion for the fine and coarse discretisations: we first generate the fine Brownian increments $\displaystyle \hat W=(\Delta W_1,\Delta W_2,\ldots,\Delta W_{N_f})$ and then use $\hat W^c=(\Delta W_1+\Delta W_2,\ldots,\Delta W_{N_f-1}+\Delta W_{N_f})$ as the coarse level's increments.

\subsubsection*{Estimated complexity and analysis}

Unless otherwise stated, the simulations used to illustrate this paper use the parameters $S_0=100, \, K=100, \, r=0.05, \, \sigma=0.20, \, T=1$.

In figure \ref{fig_PwS} we plot $\displaystyle (\mathbb{V}(\hat Y_l))$ as a function of $(1/h_l)$ in a log-log plot. We then measure the slopes for the different estimators: this gives a numerical estimate of the parameter $\beta$ in theorem~\ref{burgos:complexity_thm}. Combining this with the theorem, we get an estimated complexity of the multilevel algorithm (table \ref{table_PwS}). This gives the following results :

\begin{figure}[h!]
\centering
\caption{\label{fig_PwS} Pathwise sensitivities, European call : $\mathbb V (\hat Y_l)(l) $}
\includegraphics[width=.9\textwidth]{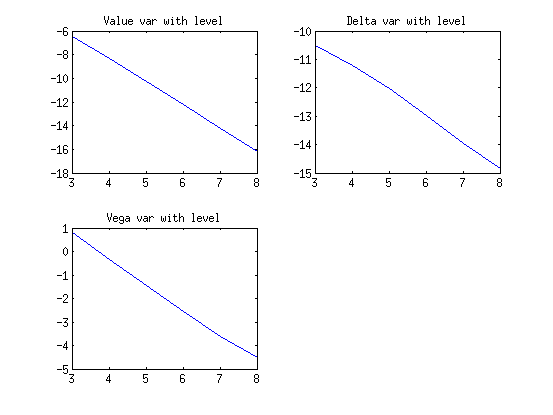}
\end{figure}

\begin{table}[h!]
\centering
\caption{\label{table_PwS} Pathwise sensitivities, European call : estimated complexity }
\begin{tabular}{|c|c|c|}
\hline \textrm{Estimator} & $\beta$ & \textrm{MLMC Complexity}\\ 
\hline \textrm{Value} & $\approx 2.0$ & $O(\epsilon^{-2})$\\ 
\hline \textrm{Delta} & $\approx 0.8$ & $O(\epsilon^{-2.2})$\\ 
\hline \textrm{Vega} & $\approx 1.0$ & $O(\epsilon^{-2}\log{\epsilon}^2)$\\ 
\hline 
\end{tabular} 
\end{table}

Giles has shown in \cite{burgos:giles07b} that $\beta\!=\!2$ for the value's estimator. For Greeks, the convergence is degraded by the discontinuity of $\frac{\partial P}{\partial S}=\mathbf{1}_{S>K}$: a fraction $O(h_l)$ of the paths has a final value $\hat S$ which is $O(h_l)$ from the discontinuity $K$. For these paths, there is a $O(1)$ probability that $\hat S_{N_f}^{(l)}$ and $\hat S_{N_c}^{(l-1)}$ are on different sides of the strike $K$, implying
%the difference between the fine and coarse Greeks' estimates will be
$ \left(
\frac{\partial P}{\partial S_{N_f}}
\frac{\partial \hat S_{N_f}}{\partial \theta}
\right)^{(l)} - \left( 
\frac{\partial P}{\partial S_{N_c}}
\frac{\partial \hat  S_{N_c}}{\partial \theta}
\right)^{(l-1)} = O(1)$. Thus $\mathbb{V}(\hat Y_l)=O(h_l)$, and $\beta=1$ for the Greeks.

\subsection{Pathwise sensitivities and Conditional Expectations}
\label{burgos:pwscondexp}

We have seen that the payoff's lack of smoothness prevents the variance of Greeks' estimators $\hat Y_l$ from decaying quickly and limits the potential benefits of the multilevel approach.
To improve the convergence speed, we can use 
%Glasserman's idea of using 
conditional expectations \cite{burgos:glasserman04}.  Instead of simulating the whole path, we stop at the penultimate step and then for every fixed set $\displaystyle \hat W=\left(\Delta W_k \right)_{k \in [1,N-1]}$, we consider the full distribution of $\left( \hat{S}_N | \hat W \right)$. 
With $a_n=a \left( \hat{S}_{N-1}(\hat W) , (N-1) h \right)$ and $b_n=b \left( \hat{S}_{N-1}(\hat W), (N-1)h \right) $, we can write
\begin{equation}
\hat{S}_N(\hat W,\hat W_N) = \hat{S}_{N-1}(\hat W)+a_n(\hat W) h + b_n(\hat W)\, \hat W_N
\end{equation}
We hence get a normal distribution for $\left( \hat{S}_N | \hat W \right)$.
\begin{equation} 
p(\hat{S}_N | \hat W) = \frac{1}{\sigma_{\hat W} \sqrt{2 \pi}} \, \exp \left( - \frac{ \left( \hat{S}_N - \mu_{\hat W} \right)^2 }{2 \sigma_{\hat W}^2}\right) \label{burgos:probadis}
\end{equation}
\begin{equation}  \nonumber \textrm{with }\left\{
    \begin{array}{ll}
        \displaystyle \mu_{\hat W}=\hat{S}_{N-1}+a \left( \hat{S}_{N-1} , (N-1) h \right) h \\
        \displaystyle \sigma_{\hat W}=b \left( \hat{S}_{N-1}, (N-1)h \right) \sqrt{h}
    \end{array}
\right.
\end{equation}
We can thus compute 
$\displaystyle \mathbb{E}\left[P\left(\hat{S}_N \right)| \hat{W} \right]$. Using the chain rule, we get
\begin{equation}
\hat{V}=\mathbb{E}\left[P(\hat S_N)\right]
=\mathbb{E}_{\hat W}\left[\mathbb{E}_{\Delta W_{N}}\left[P(\hat S_N)|\hat W\right]\right]
\approx \frac{1}{M} \sum \limits_{m=1}^M \mathbb{E}\left[P(\hat{S}_N^{(m)})|\hat{W}^{(m)}\right]
\label{burgos:eqn17}
\end{equation}

Here with $\phi$ the normal probability density function, $\Phi$  the  normal cumulative distribution functions, $\alpha=(1+r h) \hat S_{N-1}(\hat W)$ and $\beta~=~\sigma~\sqrt{h} \hat S_{N-1}(\hat W)$, we get
\begin{equation}
\mathbb{E}(P(\hat S_N) | \hat W)=  \;
\beta \,\phi \left(\frac{\alpha-K}{\beta}\right)+
(\alpha-K) \, \Phi \left( \frac{\alpha - K}{\beta} \right)
\label{burgos:formula18}
\end{equation}
This expected payoff is infinitely differentiable with respect to the input parameters. We can apply the pathwise sensitivities technique to this smooth function at time $(N-1)\,h$.
The multilevel estimator for the Greek is then
\begin{equation} \label{burgos:greekpwscondexp}
\hat Y_l = \frac{1}{N_l}\sum_1^{N_l} \left[ 
\left(\frac{\partial \hat P_f}{\partial \theta}^{(i)}\right)^{(l)} - 
\left(\frac{\partial \hat P_c}{\partial \theta}^{(i)}\right)^{(l-1)} \right]
\end{equation}
%\begin{equation}  \nonumber \textrm{with }\left\{
%    \begin{array}{ll}
%        \displaystyle \left(\frac{\partial \hat P_f}{\partial \theta}\right)^{(l)}=\frac{ \partial \hat S_{N_f-1}}{\partial \theta}\frac{\partial \mathbb{E}(P(\hat S_{N_f}) | \hat W)}{\partial S_{N_f-1}}+\frac{\partial \mathbb{E}(P(\hat S_{N_f}) | \hat W)}{\partial \theta}  \\
%        \displaystyle \left(\frac{\partial \hat P_c}{\partial \theta}\right)^{(l-1)}=\frac{\partial \hat  S_{N_c-1}}{\partial \theta}\frac{\partial \mathbb{E}(P(\hat S_{N_c}) | \hat W)}{\partial S_{N_c-1}}+\frac{\partial \mathbb{E}(P(\hat S_{N_c}) | \hat W)}{\partial \theta}
%    \end{array}
%\right.
%\end{equation}

At the fine level we use \eqref{burgos:formula18} with $h=\displaystyle h_f$ and $\displaystyle \hat W_f=(\Delta W_1,\Delta W_2,\ldots,\Delta W_{N_f-1})$ to get $\displaystyle \mathbb{E}(P(\hat S_{N_f}) | \hat W_f)$  . We then use
\begin{equation}
\displaystyle \left(\frac{\partial \hat P_f}{\partial \theta}\right)^{(l)}=\frac{ \partial \hat S_{N_f-1}}{\partial \theta}\frac{\partial \mathbb{E}(P(\hat S_{N_f}) | \hat W_f)}{\partial S_{N_f-1}}+\frac{\partial \mathbb{E}(P(\hat S_{N_f}) | \hat W_f)}{\partial \theta}
\end{equation}

At the coarse level, directly using $\mathbb{E}(P(\hat S_{N_c}) | \hat W_c)$  leads to an unsatisfactorily low convergence rate of $\mathbb V(\hat Y_l)$.
As explained in \eqref{burgos:NewMultilevelEstimator} we use a modified estimator. The idea is to include the final fine Brownian increment in the computation of the expectation over the last coarse timestep. This guarantees that the two paths will be close to one another and helps achieve better variance convergence rates.

$\hat S$ still follows a simple Brownian motion with constant drift and volatility on all coarse steps. With   $\hat W_c=(\Delta W_1+\Delta W_2,\ldots,\Delta W_{N_f-3}+\Delta W_{N_f-2})$ and given that the Brownian increment on the first half of the final step is $\Delta W_{N_f-1}$ , we get

\begin{equation} 
p(\hat{S}_{N_c} | \hat W_c,\Delta W_{N_f-1}) = \frac{1}{\sigma_{\hat W_c} \sqrt{2 \pi}} \, \exp \left( - \frac{ \left( \hat{S}_{N_c} - \mu_{\hat W_c} \right)^2 }{2 \sigma_{\hat W_c}^2}\right) \label{burgos:probadis2}
\end{equation}
\begin{equation}  \nonumber \textrm{with }\left\{
    \begin{array}{ll}
        \displaystyle \mu_{\hat W_c}=\hat{S}_{N_c-1}(\hat W_c)+a \left( \hat{S}_{N_c-1} , (N_c-1) h_c \right) h_c 
        + b \left( \hat{S}_{N_c-1} , (N_c-1) h_c \right) \Delta W_{N_f-1}\\
        \displaystyle \sigma_{\hat W_c}=b \left( \hat{S}_{N_c-1}, (N_c-1)h_c \right) \sqrt{h_c/2}
    \end{array}
\right.
\end{equation}
From this distribution we derive $\mathbb{E}\left[P(\hat{S}_{N_c})|\hat{W}_c,\Delta W_{N_f-1}\right]$, which leads to the same payoff formula as before with $\displaystyle \alpha_c=(1+r\,h_c+\sigma \Delta W_{N_f-1})\, \hat S_{N_c-1}(\hat W_c)$ and \linebreak $\displaystyle \beta_c=\sigma \sqrt{h_c}\, \hat S_{N_c-1}(\hat W_c)$. Using it as the coarse level's payoff does not introduce any bias. Using the tower property we check that it satisfies condition \eqref{burgos:TelescopeCondition},
\begin{equation}
\mathbb{E}_{\Delta W_{N_f-1}}\left[\mathbb{E}\left[P(\hat{S}_{N_c})|\hat{W}_c,\Delta W_{N_f-1}\right]|\hat{W}_c\right]
=\mathbb{E}\left[P(\hat{S}_{N_c})|\hat{W}_c\right] \nonumber
\end{equation}

\subsubsection*{Estimated complexity and analysis}

Our numerical experiments show the benefits of the conditional expectation technique on the European call:
\begin{figure}[!h]
\centering
\caption{\label{fig_PwS_CondExp_Euro} Pathwise sensitivities and conditional expectations, \\European call : $\mathbb V (\hat Y_l)(l) $}
\includegraphics[width=.9\textwidth]{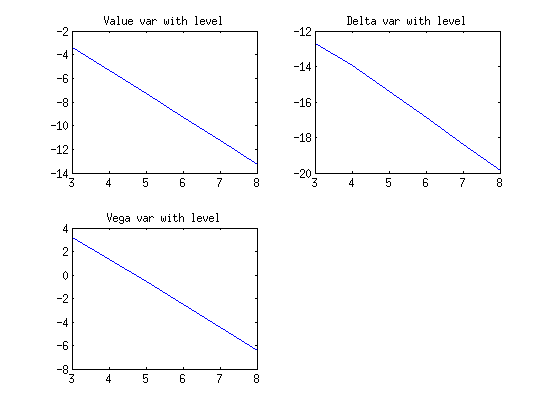}
\end{figure}

\begin{table}[!h]
\centering
\caption{\label{table_PwS_CondExp_Euro} Pathwise sensitivities and conditional expectations, \\European call : estimated complexity }
\begin{tabular}{|c|c|c|}
\hline \textrm{Estimator} & $\beta$ & \textrm{MLMC Complexity}\\ 
\hline \textrm{Value} & $\approx 2.0$ & $O(\epsilon^{-2})$\\ 
\hline \textrm{Delta} & $\approx 1.5$ & $O(\epsilon^{-2})$\\ 
\hline \textrm{Vega} & $\approx 2.0$ & $O(\epsilon^{-2})$\\ 
\hline 
\end{tabular} 
\end{table}
\newpage 
A fraction $O(\sqrt{h_l})$ of the paths arrive in the area around the strike where the conditional expectation $\displaystyle\frac{\partial \mathbb{E}(P(\hat S_{N}) | \hat W)}{ \partial \hat S_{N_f-1}}$ is neither close to 0 nor 1. In this area, its slope is $O(h_l^{-1/2})$. The coarse and fine paths differ by $O(h_l)$, we thus have $O(\sqrt{h_l})$ difference between the coarse and fine Greeks' estimates. Reasoning as in \cite{burgos:giles07b} we get  $\mathbb{V}_{\hat W}(\mathbb{E}_{\Delta W_N}(...|\hat W))=O(h_l^{3/2})$ for the Greeks' estimators. This is the convergence rate observed for $\delta$; the higher convergence rate of $\nu$ is not explained yet by this rough analysis and will be investigated in our future research.

The main limitation of this approach is that in many situations it leads to complicated integral computations. Path splitting, to be discussed next, may represent a useful numerical approximation to this technique.

\subsection{Split pathwise sensitivities}
\label{burgos:splitpws}
This technique is based on the previous one. The idea is to avoid the tricky computation of 
$\mathbb{E}\left[P(\hat S_{N_f}) | \hat W_f\right]$ and
$\mathbb{E}\left[P(\hat{S}_{N_c})|\hat{W}_c,\Delta W_{N_f-1}\right]$. We get numerical estimates of these values by ``splitting" every path simulation on the final timestep.

At the fine level: for every simulated path $\hat W_f=(\Delta W_1,\Delta W_2,\ldots,\Delta W_{N_f-1})$, we simulate a set of $d$ final increments
$(\Delta W_{N_f}^{(i)})_{i \in [1,d]}$ which we average to get
\begin{equation}
\displaystyle \mathbb{E}\left[P(\hat S_{N_f}) | \hat W_f\right] \approx 
\frac{1}{d}\, \sum_{i=1}^{d}P(\hat S_{N_f}(\hat W_f,  \Delta W_{N_f}^{(i)})) 
\end{equation}

At the coarse level we use $\hat W_c=(\Delta W_1+\Delta W_2,\ldots,\Delta W_{N_f-3}+\Delta W_{N_f-2})$. As before (still assuming a constant drift and volatility on the final coarse step), we improve the convergence rate of $\mathbb V(\hat Y_l)$ by reusing $\Delta W_{N_f-1}$ in our estimation of $\mathbb{E}\left[P(\hat{S}_{N_c})|\hat{W}_c\right]$ . We can do so by constructing the final coarse increments as
$(\Delta W_{N_c}^{(i)})_{i \in [1,d]}=(\Delta W_{N_f-1}+(\Delta W_{N_f}^{(i)}))_{i \in [1,d]}$ and using these to estimate
\begin{equation}
\displaystyle \mathbb{E}(P(\hat S_{N_c}) | \hat W_c) =
\mathbb{E}\left[P(\hat{S}_{N_c})|\hat{W}_c,\Delta W_{N_f-1}\right] \nonumber
\approx 
\frac{1}{d}\, \sum_{i=1}^{d}P(\hat S_{N_c}(\hat{W}_c,\Delta W_{N_c}^{(i)}))
\end{equation}
%To get the Greeks we use
%\begin{equation} \label{burgos:greeksplitpws}
%\hat Y_l = \frac{1}{N_l}\sum \left[ \left(\frac{\partial \hat S_{N_f-1}}{\partial \theta}\frac{\partial \bar P_f}{\partial S_{N_f-1}} \right)^{(l)} - \left( \frac{\partial \hat  S_{N_c-1}}{\partial \theta}\frac{\partial \bar P_c}{\partial S_{N_c-1}} \right)^{(l-1)} \right]
%\end{equation}
%\begin{equation}\nonumber% \textrm{with }\left\{
%%	\begin{array}{ll}
%        \displaystyle \frac{\partial \bar P_f}{\partial S_{N_f-1}}
%		=        
%        \frac{1}{d} \, \sum_{i=1}^{d}
%		 \left( \frac{\partial \hat S_{N_f}}{\partial \hat S_{N_{f-1}}} \,
%		 \frac{\partial P}{\partial S_{N_f}} \right)^{(i)}
%		\qquad 
%       % \\
%        \displaystyle \frac{\partial \bar P_c}{\partial S_{N_c-1}}
%        =
%		\frac{1}{d} \, \sum_{i=1}^{d}
%		 \left( \frac{\partial \hat S_{N_c}}{\partial \hat S_{N_{c-1}}} \,
%		 \frac{\partial P}{\partial S_{N_c}} \right)^{(i)}
%%    \end{array}\right.
%\end{equation}

To get the Greeks, we simply compute the corresponding pathwise sensitivities.

\subsubsection*{Estimated complexity and choice of the number of splittings}

\begin{figure}[h!]
\centering
\caption{\label{fig_PwS_Split} Pathwise sensitivities and path splitting, \\European call : $\mathbb V (\hat Y_l)(l) $}
\includegraphics[width=.9\textwidth]{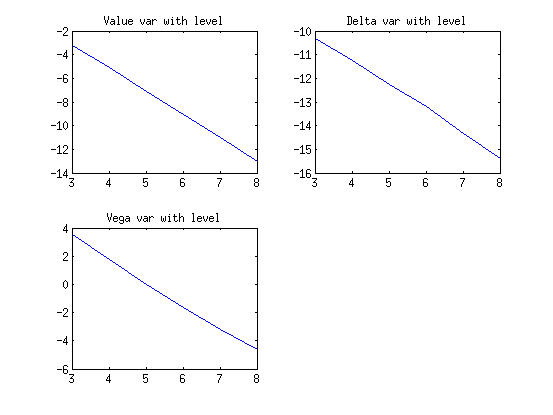}
\end{figure}
\begin{table}[h!]
\centering
\caption{\label{table_PwS_Split} Pathwise sensitivities and path splitting, \\European call : estimated complexity }
\begin{tabular}{|c|c|c|c|}
\hline \textrm{Estimator} &$d $& $\beta$ & \textrm{MLMC Complexity}\\ 
\hline \textrm{Value}
					  & $10 $& $\approx 2.0$ & $O(\epsilon^{-2})$\\
					  & $500 $& $\approx 2.0$ & $O(\epsilon^{-2})$\\
\hline \textrm{Delta} 
					  & $10 $& $\approx 1.0$ & $O(\epsilon^{-2}(\log \epsilon)^2)$\\
					  & $500 $& $\approx 1.5$ & $O(\epsilon^{-2})$\\
\hline \textrm{Vega} 
					  & $10 $& $\approx 1.6$ & $O(\epsilon^{-2})$\\
					  & $500 $& $\approx 2.0$ & $O(\epsilon^{-2})$\\
\hline 
\end{tabular} 
\end{table}
As expected this method yields higher values of $\beta$ than simple pathwise sensitivities: the convergence rates increase and tend to the rates offered by conditional expectations as $d$ increases and the approximation gets more precise.

Taking a constant number of splittings $d$ for all levels is actually not optimal; For Greeks, we can write the variance of the estimator as
\begin{align}\nonumber
\displaystyle \mathbb{V}(\hat Y_l)=
&\displaystyle \frac{1}{N_l}\mathbb{V}_{\hat W_f}(\mathbb{E}(\left(\frac{\partial \hat P_f}{\partial \theta}\right)^{(l)} - \left(\frac{\partial \hat P_c}{\partial \theta}\right)^{(l-1)} |\hat W_f))\\
&\displaystyle +\frac{1}{N_l\,d}\mathbb{E}_{\hat W_f}(\mathbb{V}(\left(\frac{\partial \hat P_f}{\partial \theta}\right)^{(l)} - \left(\frac{\partial \hat P_c}{\partial \theta}\right)^{(l-1)} |\hat W_f))
\end{align}
As explained in section \ref{burgos:pwscondexp} we have $\mathbb{V}_{\hat W_f}(\mathbb{E}(...|\hat W_f))=O(h_l^{3/2})$ for the Greeks. We also have $\mathbb{E}_{\hat W_f}(\mathbb{V}(...|\hat W_f))=O(h_l)$ for similar reasons. We optimise the variance at a fixed computational cost by choosing $d$ such that the two terms of the sum are of similar order. Taking $d=O(h_l^{-1/2})$ is therefore optimal.

%Problems arise with nonsmooth payoffs. As is evident in \eqref{burgos:greeksplitpws}, the method purely relies on pathwise sensitivities computations and can only be used in cases where simple pathwise sensitivities are applicable. To address this limitation, Giles introduced the Vibrato Monte Carlo \cite{burgos:giles09}.  This hybrid method combines pathwise sensitivities and the Likelihood Ratio Method. 

\subsection{Vibrato Monte Carlo}

Since the previous method uses pathwise sensitivity analysis, it is not
applicable when payoffs are discontinuous. To address this limitation, we 
use the Vibrato Monte Carlo method introduced by Giles \cite{burgos:giles09}.
This hybrid method combines pathwise sensitivities and the Likelihood Ratio Method. 

We consider again equation \eqref{burgos:eqn17}. We now use the Likelihood Ratio Method on the last timestep and with the notations of section \ref{burgos:pwscondexp} we get
\begin{equation}
 \frac{\partial \hat V}{\partial \theta} =
\mathbb{E}_{\hat W} \left[ \mathbb{E}_{\Delta W_N} \left[ P \left( \hat{S}_N\right) \frac{\partial (\log p(\hat{S}_N | \hat W))}{\partial \theta}  | \hat W \right] \right]
\end{equation}
We can write $p(\hat{S}_N | \hat W))$ as $p( \mu_{\hat W},\sigma_{\hat W}) $. This leads to the estimator
\begin{eqnarray}
\nonumber\displaystyle  \frac{\partial \hat V}{\partial \theta} \approx
\frac{1}{N_l} \sum\limits_{m=1}^{N_l} & ( \displaystyle\frac{\partial \mu_{\hat W^{(m)}}}{\partial \theta} \mathbb{E}_{\Delta W_N} \left[ P \left( \hat{S}_N \right) \frac{\partial (\log p)}{\partial \mu_{\hat W}} |\hat W^{(m)} \right] \\ 
 \displaystyle & + \displaystyle\frac{\partial \sigma_{\hat W^{(m)}}}{\partial \theta} \mathbb{E}_{\Delta W_N} \left[ P \left( \hat{S}_N \right) \frac{\partial (\log p)}{\partial \sigma_{\hat W}} |\hat W^{(m)} \right]  )
 \label{burgos:vmcestimator}
\end{eqnarray}
We compute $\displaystyle \frac{\partial \mu_{\hat W^{(m)}}}{\partial \theta}$ and $\displaystyle \frac{\partial \sigma_{\hat W^{(m)}}}{\partial \theta}$ with pathwise sensitivities. 
\\With $\displaystyle \hat{S}_N^{(m,i)}=\hat{S}_N(\hat W^{(m)},\Delta W_N^{(i)}) $, we substitute the following estimators into \eqref{burgos:vmcestimator}

\begin{equation} \label{burgos:vmcestimator2}
\left\{\begin{array}{ll}
\displaystyle \mathbb{E}_{\Delta W_N} \left[ P \left( \hat{S}_N \right) \frac{\partial (\log p)}{\partial \mu_{\hat W}} |\hat W^{(m)} \right] \approx
\frac{1}{d} \sum\limits_{i=1}^d \left( P \left( \hat{S}_N^{(m,i)} \right) \frac{\hat{S}_N^{(m,i)} - \mu_{\hat W^{(m)}}}{\sigma_{\hat W^{(m)}}^2} \right)
\\
\displaystyle 
\mathbb{E}_{\Delta W_N} \left[ P \left( \hat{S}_N \right) \frac{\partial (\log p)}{\partial \sigma_{\hat W}} |\hat W^{(m)} \right] \approx
\frac{1}{d} \sum\limits_{i=1}^d P \left( \hat{S}_N^{(m,i)}\right) \left( -\frac{1}{\sigma_{\hat W^{(m)}}} + \frac{\left( \hat{S}_N^{(m,i)} - \mu_{\hat W^{(m)}}\right)^2}{\sigma_{\hat W^{(m)}}^3}\right)
\end{array}
\right.\nonumber
\end{equation}
%which we substitute in \eqref{burgos:vmcestimator} to get the Vibrato Monte Carlo estimator of $\displaystyle \frac{\partial  \hat V}{\partial \theta}$.

In a multilevel setting: at the fine level we can use \eqref{burgos:vmcestimator} directly.
At the coarse level, for the same reasons as in section  \ref{burgos:splitpws}, we reuse the fine brownian increments to get efficient estimators. We take
\begin{equation}
\left\{\begin{array}{ll}
\displaystyle \hat W_c=(\Delta W_1+\Delta W_2,\ldots,\Delta W_{N_f-3}+\Delta W_{N_f-2})\\
\displaystyle (\Delta W_{N_c}^{(i)})_{i \in [1,d]}=(\Delta W_{N_f-1}+(\Delta W_{N_f}^{(i)}))_{i \in [1,d]}
\end{array}
\right.
\end{equation}
We use the chain rule to verify that condition \eqref{burgos:TelescopeCondition} is verified on the last coarse step. With the notations of equation \eqref{burgos:probadis2} we derive the following estimators
\begin{equation}
%\left\{
\begin{array}{ll}
\displaystyle \mathbb{E}_{\Delta W_{N_c}} \left[ P \left( \hat{S}_{N_c} \right) \frac{\partial (\log p_c)}{\partial \mu_{\hat W_c}} |\hat W_c^{(m)} \right] 
& = \displaystyle
\mathbb{E}
%_{\Delta W_{N_f-1}}
\left[\mathbb{E}\left[P \left( \hat{S}_{N_c} \right) \frac{\partial (\log p_c)}{\partial \mu_{\hat W_c}}|\hat{W}_c^{(m)},\Delta W_{N_f-1}\right]|\hat{W}_c^{(m)}\right]
\\
&  \approx
\displaystyle \frac{1}{d} \sum\limits_{i=1}^d \left( P \left( \hat{S}_{N_c}^{(m,i)} \right) \frac{\hat{S}_{N_c}^{(m,i)} - \mu_{\hat W_c^{(m)}}}{\sigma_{\hat W_c^{(m)}}^2} \right)\\
\displaystyle \mathbb{E}_{\Delta W_{N_c}} \left[ P \left( \hat{S}_{N_c} \right) \frac{\partial (\log p)}{\partial \sigma_{\hat W_c}} |\hat W_c^{(m)} \right] &= \displaystyle
\mathbb{E} 
%_{\Delta W_{N_f-1}}
\left[\mathbb{E}\left[ P \left( \hat{S}_{N_c} \right) \frac{\partial (\log p)}{\partial \sigma_{\hat W_c}} |\hat{W}_c^{(m)},\Delta W_{N_f-1}\right]|\hat{W}_c^{(m)}\right]\\
& \approx
\displaystyle \frac{1}{d} \sum\limits_{i=1}^d P \left( \hat{S}_{N_c}^{(m,i)}\right) \left( -\frac{1}{\sigma_{\hat W_c^{(m)}}} + \frac{\left( \hat{S}_{N_c}^{(m,i)} - \mu_{\hat W_c^{(m)}}\right)^2}{\sigma_{\hat W_c^{(m)}}^3}\right)
\end{array}
%\right.
\end{equation}

\subsubsection*{Estimated complexity}
Our numerical experiments show the following convergence rates for $d=10$:

\begin{figure}[h!]
\centering
\caption{\label{fig_VMC_Euro} Vibrato Monte Carlo, European call : $\mathbb V (\hat Y_l)(l) $}
\includegraphics[width=.9\textwidth]{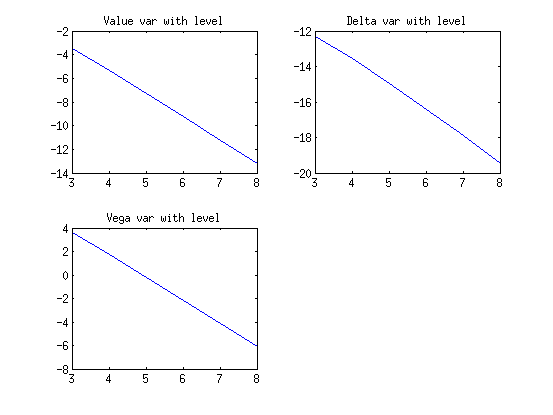}
\end{figure}

\begin{table}[h!]
\centering
\caption{\label{table_VMC_Euro} Vibrato Monte Carlo, European call : estimated complexity }
\begin{tabular}{|c|c|c|}
\hline \textrm{Estimator} & $\beta$ & \textrm{MLMC Complexity}\\ 
\hline \textrm{Value} & $\approx 2.0$ & $O(\epsilon^{-2})$\\ 
\hline \textrm{Delta} & $\approx 1.5$ & $O(\epsilon^{-2})$\\ 
\hline \textrm{Vega} & $\approx 2.0$ & $O(\epsilon^{-2})$\\ 
\hline 
\end{tabular} 
\end{table}

\newpage
As in section \ref{burgos:splitpws}, this is an approximation of the conditional expectation technique, and so getting the same convergence rates was expected.

\section{European digital call}
\label{burgos:3}
The European digital call's payoff is $P=\mathbf{1}_{S_T>K}$. The discontinuity of the payoff makes the computation of Greeks more challenging. We cannot apply pathwise sensitivities, and so we use conditional expectations or Vibrato Monte Carlo.
\subsection{Pathwise sensitivities and conditional expectations}

With the same notation as in section \ref{burgos:pwscondexp} we compute the conditional expectations of the digital call's payoff.
\begin{equation}
\nonumber \displaystyle \mathbb{E}(P(\hat S_{N_f}) | \hat W)=\Phi \left( \frac{\alpha-K}{\beta}\right)
\qquad \displaystyle \mathbb{E}(P(\hat S_{N_c}) | \hat W_c,\Delta W_{N_f-1})=\Phi \left( \frac{\alpha_c-K}{\beta_c}\right)
\end{equation}

The simulations give figure \ref{fig_PwS_CondExp_Dig} and table \ref{table_PwS_CondExp_Dig}.
\begin{figure}[h!]
\centering
\caption{\label{fig_PwS_CondExp_Dig} Pathwise sensitivities and conditional expectations, \\digital call: $\mathbb V (\hat Y_l)(l) $}
\includegraphics[width=.9\textwidth]{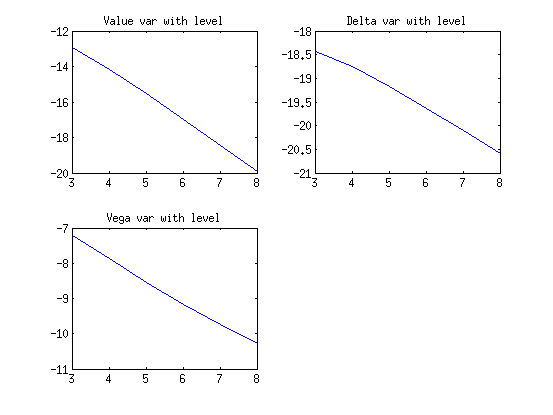}
\end{figure}

\begin{table}[h!]
\centering
\caption{\label{table_PwS_CondExp_Dig} Pathwise sensitivities and conditional expectations, \\digital call : estimated complexity }
\begin{tabular}{|c|c|c|}
\hline \textrm{Estimator} & $\beta$ & \textrm{MLMC Complexity}\\ 
\hline \textrm{Value} & $\approx 1.4$ & $O(\epsilon^{-2})$\\ 
\hline \textrm{Delta} & $\approx 0.5$ & $O(\epsilon^{-2.5})$\\ 
\hline \textrm{Vega} & $\approx 0.6$ & $O(\epsilon^{-2.4})$\\ 
\hline 
\end{tabular} 
\end{table}
\newpage
\subsection{Vibrato Monte Carlo}
The Vibrato technique can be applied in the same way as with the European call. We get figure \ref{fig_VMC_Dig} and table \ref{table_VMC_Dig}.

\begin{figure}[h!]
\centering
\caption{\label{fig_VMC_Dig} Vibrato Monte Carlo, digital call : $\mathbb V (\hat Y_l)(l) $}
\includegraphics[width=.9\textwidth]{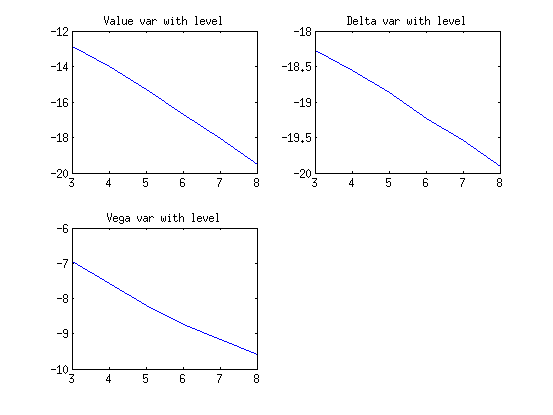}
\end{figure}

\begin{table}[h!]
\centering
\caption{\label{table_VMC_Dig} Vibrato Monte Carlo, digital call : estimated complexity }
\begin{tabular}{|c|c|c|}
\hline \textrm{Estimator} & $\beta$ & \textrm{MLMC Complexity}\\ 
\hline \textrm{Value} & $\approx 1.3$ & $O(\epsilon^{-2})$\\ 
\hline \textrm{Delta} & $\approx 0.3$ & $O(\epsilon^{-2.7})$\\ 
\hline \textrm{Vega} & $\approx 0.5$ & $O(\epsilon^{-2.5})$\\ 
\hline 
\end{tabular} 
\end{table}

\subsection{Analysis}

The analysis presented in section \ref{burgos:pwscondexp} explains why we expected $\beta=3/2$ for the value's estimator. 

A fraction $O(\sqrt{h})$ of all paths arrive in the area around the payoff where \linebreak
 $(\partial \mathbb{E}(P(\hat S_{N}) | \hat W)/ \partial \hat S_{N-1})$ is not close to 0 ; there its derivative is $O(h_l^{-1})$ and we have $|\hat S_{N_f}- \hat S_{N_c}|=O(h_l)$. For these paths, we thus have $O(1)$ difference between the fine and coarse Greeks' estimates. This explains the experimental $\beta \approx 1/2$.

\section{European lookback call}
\label{burgos:4}
The lookback call's value depends on the values that the asset takes before expiry. Its payoff is  
$\displaystyle P(T)= (S_T-\min_{t \in [0,T]}(S_t)) $.

As explained in \cite{burgos:giles07b}, the natural discretisation 
$\displaystyle \hat P = ( \hat{S}_N-\min_n \hat S_n )$
 is not satisfactory. To regain good convergence rates, we approximate the behaviour within each fine timestep $[t_n,t_{n+1}] $ of width $h_f$ as a simple Brownian motion with constant drift $a_n^f$ and volatility $b_n^f$ conditional on the simulated values $\hat S_n^f$ and $\hat S_{n+1}^f$. As shown in \cite{burgos:glasserman04} we can then simulate the local minimum
\begin{equation}
\hat S_{n,min}^f=\frac{1}{2}\left(
\hat S_n^f+ \hat S_{n+1}^f - \sqrt{\left(\hat S_{n+1}^f- \hat S_n^f \right)^2-2 (b_n^f)^2 h_f \log U_n}
\right)
\end{equation}
with $U_n$ a uniform random variable on $[0,1]$. We define the fine level's payoff this way choosing $b_n^f=b(\hat S_n^f,t_n)$ and considering the minimum over all timesteps to get the global minimum of the path.

At the coarse level we still consider a simple Brownian motion on each timestep of width $h_c=2 h_f$. To get high strong convergence rates, we reuse the fine increments by defining a midpoint value for each step
\begin{equation}
\hat S_{n+1/2}^c=\frac{1}{2} \displaystyle \left(\hat S_n^c+\hat S_{n+1}^c-b_n^c (\Delta W_{n+1} - \Delta W_{n+1/2})\right)
\end{equation}
Where $(\Delta W_{n+1} - \Delta W_{n+1/2})$ is the difference of the corresponding fine Brownian increments on  $[t_{n+1/2},t_{n+1}]$ and $[t_n,t_{n+1/2}]$. Conditional on this value, we then define the minimum over the whole step as the minimum of the minimum over each half step, that is
\begin{align}
\nonumber \hat S_{n,min}^c=\min &\left[
\frac{1}{2}\left(
\hat S_n^c+ \hat S_{n+1/2}^c - \sqrt{\left(\hat S_{n+1/2}^c- \hat S_n^c \right)^2- (b_n^c)^2 h_c \log U_{1,n}}
\right), \right.\\
&\; \left. \frac{1}{2}\left(
\hat S_{n+1/2}^c+ \hat S_{n+1}^c - \sqrt{\left(\hat S_{n+1}^c- \hat S_{n+1/2}^c \right)^2- (b_n^c)^2 h_c \log U_{2,n}}
\right)
 \right]
\end{align}
where $U_{1,n}$ and $U_{2,n}$ are the values we sampled to compute the minima of the corresponding  timesteps at the fine level. Once again we use the tower property to check that condition \eqref{burgos:TelescopeCondition} is verified and that this coarse-level estimator is adequate.

\subsection{pathwise sensitivities}
Using the treatment described above, we can then apply straighforward pathwise sensitivities to compute the multilevel estimator.  This gives the following results:
%\begin{equation}
%\displaystyle \hat Y_l=
%\left(\frac{\partial \hat S_N^f}{\partial \theta} -\frac{\partial \displaystyle \min_n \hat S_{n,min}^f}{\partial \theta}\right)^{(l)}
%-
%\left(\frac{\partial \hat S_N^c}{\partial \theta} -\frac{\partial \displaystyle \min_n \hat S_{n,min}^c}{\partial \theta}\right)^{(l-1)}
%\end{equation} 
%We get

\begin{figure}[h!]
\centering
\caption{\label{fig_Lookback} Pathwise sensitivities, lookback call : $\mathbb V (\hat Y_l)(l) $}
\includegraphics[width=.9\textwidth]{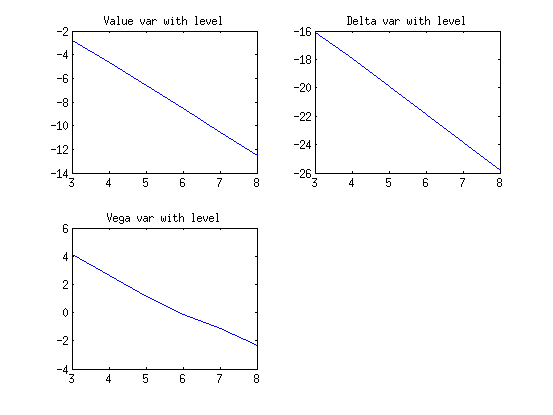}
\end{figure}

\begin{table}[h!]
\centering
\caption{\label{table_Lookback} Pathwise sensitivities, lookback call : estimated complexity }
\begin{tabular}{|c|c|c|}
\hline \textrm{Estimator} & $\beta$ & \textrm{MLMC Complexity}\\ 
\hline \textrm{Value} & $\approx 1.9$ & $O(\epsilon^{-2})$\\ 
\hline \textrm{Delta} & $\approx 1.9$ & $O(\epsilon^{-2})$\\ 
\hline \textrm{Vega} & $\approx 1.3$ & $O(\epsilon^{-2})$\\ 
\hline 
\end{tabular} 
\end{table}

Giles has proved that for the value's estimator, $\beta=2-\eta $ for all $\eta >0$. In the Black~\& Scholes model, we can prove that $V=S_0 \, \delta$. We therefore expected $\beta \approx 2$ for $\delta$ too.  The strong convergence speed of $\nu$'s estimator cannot be derived that easily and will be analysed in our future research.

\subsection{Conditional Expectations, path splitting or Vibrato Monte Carlo}
Unlike the regular call option, the payoff of the lookback call is perfectly smooth and so therefore there is no benefit from using conditional expectations and associated methods.

\section{European barrier call}
\label{burgos:5}
Barrier options are contracts which are activated or deactivated when the underlying asset $S$ reaches a certain barrier value $B$. We consider here the down-and-out call for which the payoff can be written as
\begin{equation}
\displaystyle P=(S_T-K)^+\, \mathbf{1}_{\displaystyle \min_{t \in [0,T]}( S_t)>K}
\end{equation}
Both the naive estimators and the approach used with the lookback call are unsatisfactory here: the discontinuity induced by the barrier results in a higher variance than before. Therefore we use the approach developed in \cite{burgos:giles07b} where we compute the probability $p_n$ that the minimum of the interpolant crosses the barrier within each timestep. This gives the conditional expectation of the payoff conditional on the Brownian increments of the fine path:
\begin{equation}
\hat P^{f}=(\hat S_{N_f}^{f}-K)^+\, \displaystyle \prod_{n=0}^{N_f-1} \left(1-\hat p_n^{f} \right)\\
\end{equation}
with
\begin{equation}
\hat p_n^{f}=
\exp \left(\frac{-2 (\hat S_n^{f}-B)^{+}(\hat S_{n+1}^{f}-B)^{+}}{(b_n^{f})^2\, h_f}  \right)\nonumber
\end{equation}
At the coarse level we define the payoff similarly: we first simulate a midpoint value $\hat S_{n+1/2}^c$ as before and then define $\hat p_n^{c}$ the probability of not hitting $B$ in $[t_n,t_{n+1} ]$, that is the probability of not hitting $B$ in $[t_n,t_{n+1/2}] $ and $[t_{n+1/2},t_{n+1}] $. Thus
\begin{equation}
\hat P^{c}=(\hat S_{N_c}^{c}-K)^+\, \displaystyle \prod_{n=0}^{N_c-1} \left(1-\hat p_n^{c} \right)
=
(\hat S_{N_c}^{c}-K)^+\, \displaystyle \prod_{n=0}^{N_c-1} \left((1-\hat p_{n,1})(1-\hat p_{n,2}) \right)
\end{equation}
with
\begin{equation}\left\{ \begin{array}{ll}
\nonumber \displaystyle \hat p_{n,1}=\exp \left(\frac{-2 (\hat S_n^{c}-B)^{+}(\hat S_{n+1/2}^{c}-B)^{+}}{(b_n^{c})^2\, h_f}  \right)
\\
\nonumber  \displaystyle \hat p_{n,2}=\exp \left(\frac{-2 (\hat S_{n+1/2}^{c}-B)^{+}(\hat S_{n+1}^{c}-B)^{+}}{(b_n^{c})^2\, h_f}  \right)
 \end{array} \right.
\end{equation}

\subsection{Pathwise sensitivities}
\label{burgos:51}
The multilevel estimators 
$\displaystyle
\hat Y_l=\left( \hat P^f\right)^{(l)} - \left(\hat P^c \right)^{(l-1)}
%(\hat S_{N_f}^{f}-K)^+\,\displaystyle \prod_{n=0}^{N_f-1} \left(1-\hat p_n^{f} \right)-
%(\hat S_{N_c}^{c}-K)^+\, \displaystyle \prod_{n=0}^{N_c-1} \left((1-\hat p_{n,1})(1-\hat p_{n,2}) \right)
$
are Lipschitz with respect to all $\displaystyle (\hat S_n^f)_{n=1 \ldots N_f}$ and $\displaystyle (\hat S_n^c)_{n=1 \ldots N_c}$, so we can use pathwise sensitivities to compute the Greeks.
% For example at the fine level
%\begin{equation}
%\frac{\partial \hat P^f}{\partial \theta}=
%\mathbf{1}_{\hat S_{N_f}^{f}>K}\, \frac{\partial \hat S_{N_f}^{f}}{\partial \theta} \displaystyle 
%\prod_{n=0}^{N_f-1} \left(1-\hat p_n^{f} \right)-
%(\hat S_{N_f}^{f}-K)^+\,
%\sum_n  \left(
%\frac{\partial \hat p_n^f}{\partial \theta}
%\prod_{k \neq n}(1-\hat p_k^f)
% \right)
%\end{equation}
%with
%\begin{equation}
%\frac{\partial \hat p_n^f}{\partial \theta}=
%\frac{\partial \hat p_n^f}{\partial \hat S_n^f}
%\frac{\partial \hat S_n^f}{\partial \theta}
%+
%\frac{\partial \hat p_n^f}{\partial \hat S_{n+1}^f}
%\frac{\partial \hat S_{n+1}^f}{\partial \theta}
%\end{equation}
%The computations at the coarse level are done similarly. 
Our numerical simulations give

\begin{figure}[h!]
\centering
\caption{\label{fig_Barrier_PwS} Pathwise sensitivities, barrier call : $\mathbb V (\hat Y_l)(l) $}
\includegraphics[width=.9\textwidth]{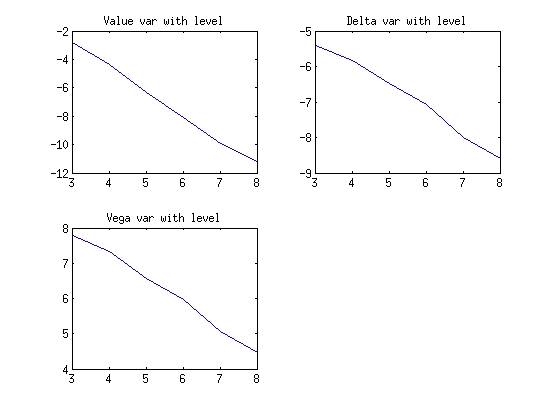}
\end{figure}

\begin{table}[h!]
\centering
\caption{\label{table_Barrier_PwS} Pathwise sensitivities, barrier call : estimated complexity }
\begin{tabular}{|c|c|c|}
\hline \textrm{Estimator} & $\beta$ & \textrm{MLMC Complexity}\\ 
\hline \textrm{Value} & $\approx 1.6$ & $O(\epsilon^{-2})$\\ 
\hline \textrm{Delta} & $\approx 0.6$ & $O(\epsilon^{-2.4})$\\ 
\hline \textrm{Vega} & $\approx 0.6$ & $O(\epsilon^{-2.4})$\\ 
\hline 
\end{tabular} 
\end{table}
\newpage
Giles proved $\beta =\frac{3}{2}-\eta$ ($\eta >0$) for the value's estimator. We are currently working on a numerical analysis supporting the observed convergence rates for the Greeks.

\subsection{Conditional Expectations}
The low convergence rates observed in the previous section come from from both the discontinuity at the barrier and from the lack of smoothness of the call around $K$. To address the latter, we can use the techniques described in section \ref{burgos:1}.  Since path splitting and Vibrato Monte Carlo offer rates that are at best equal to those of conditional expectations, we implement conditional expectations to see the maximum benefits we can get.

Computing conditional expectations is slightly trickier than in section \ref{burgos:2}. We must indeed take into account the probability that the path will hit the barrier $B$ during the final timestep. Reusing the notations of part \ref{burgos:pwscondexp} and defining

\begin{equation}
\left\{ \begin{array}{l}
%\displaystyle\alpha =(1+r h)\,\hat S_{N-1}\\
%\displaystyle\beta = \sigma \sqrt{h}\hat S_{N-1}\\
\displaystyle\tilde{\alpha}=2B+(-1+rh)\, \hat S_{N-1}(\hat{W})\\
\displaystyle L=\max(K,B)\\
\displaystyle D(\sigma,\hat{S}_{N-1})=\exp \left(\frac{2r (B-\hat S_{N-1})}{\sigma^2 \hat S_{N-1}}\right)\\
\end{array}\right.
\end{equation}
at the fine level we get
\begin{align}
\mathbb{E}(P(\hat S_N|\hat W))=&
(\alpha -K) \Phi \left( \frac{\alpha-L}{\beta}\right)+
\frac{\beta}{\sqrt{2 \pi}} 
\exp \left(-\frac{(L-\alpha)^2}{2\beta^2}\right)\\
\nonumber&-D(\sigma,\hat{S}_{N-1}) 
\left[(\tilde \alpha-K) \Phi 
\left( \frac{\tilde \alpha-L}{\beta}\right)
+\frac{\beta}{\sqrt{2 \pi}}\exp\left(-\frac{(L-\tilde\alpha)^2}{2\beta^2}\right)\right]
\end{align}
As before we then adapt this formula to the coarse level to compute 
\begin{equation}
\mathbb{E}(P(\hat S_{N_c}) | \hat W_c)=
\mathbb{E}(\mathbb{E}(P(\hat S_{N_c}) | \hat W_c, \Delta W_{N_f-1})| \hat W_c)
\end{equation}
Doing so actually leads to long impractical formulae, especially when computing the Greeks. The idea of the conditional expectation technique is to smoothen the payoff. We can quickly estimate the method's maximum benefits by replacing the true payoff by a smooth Lipschitz approximation: this introduces a bias but also eliminates all the problems due to the lack of regularity around the strike.

For example we can replace the payoff 
$
\displaystyle P=(S_T-K)^+\, \mathbf{1}_{\displaystyle \min_{t \in [0,T]}( S_t)>K}
$
by the smooth approximation
\begin{equation}
\displaystyle \tilde P=\left(\frac{\beta}{\sqrt{2\pi}}\,\exp\left(-\frac{(K-S_T)^2}{2 \beta^2}\right)
+(S_T-K) \Phi \left(\frac{S_T-K}{\beta}\right)
\right)
\, \mathbf{1}_{\displaystyle \min_{t \in [0,T]}( S_t)>K}
\end{equation}
where $ \displaystyle \beta = \sigma \sqrt{h^*} S_T$ for some arbitrary $h^*$ that controls the width of the smoothing. For example we take $h^*=1/64$ and we obtain the following results:
\begin{figure}[h!]
\centering
\caption{\label{fig_Barrier_PwS_Smoothing} Pathwise sensitivities and payoff smoothing, \\barrier call : $\mathbb V (\hat Y_l)(l) $}
\includegraphics[width=.9\textwidth]{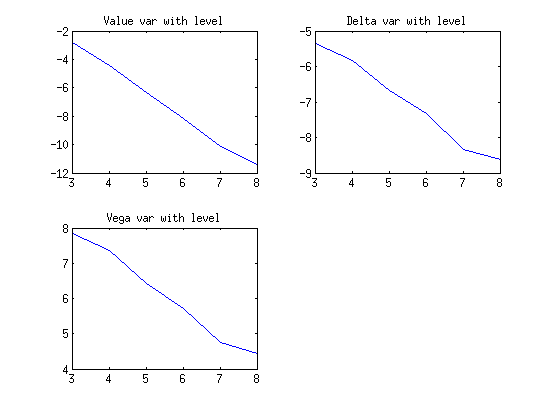}
\end{figure}

\begin{table}[h!]
\centering
\caption{\label{table_Barrier_PwS_Smoothing} Pathwise sensitivities and payoff smoothing, \\barrier call : estimated complexity }
\begin{tabular}{|c|c|c|}
\hline \textrm{Estimator} & $\beta$ & \textrm{MLMC Complexity}\\ 
\hline \textrm{Value} & $\approx 1.7$ & $O(\epsilon^{-2})$\\ 
\hline \textrm{Delta} & $\approx 0.7$ & $O(\epsilon^{-2.3})$\\ 
\hline \textrm{Vega} & $\approx 0.7$ & $O(\epsilon^{-2.3})$\\ 
\hline 
\end{tabular} 
\end{table}

We see in figure \ref{fig_Barrier_PwS_Smoothing} and table \ref{table_Barrier_PwS_Smoothing} that the maximum benefits of these techniques are only marginal. The barrier appears to be responsible for most of the variance of the multilevel estimators.

\subsection{Non-constant timestepping}

As illustrated in figure \ref{fig_Barrier_B}, the level at which $\mathbb V (\hat Y_l)$ reaches its asymptotic convergence speed depends on the value of $B$. When $B$ is far from $S_0$, the regime appears quickly (figure \ref{fig_Barrier_B:B85}), when $B$ gets closer to $S_0$, it takes longer (figure \ref{fig_Barrier_B:B95}). Practically this can be a problem when $B \approx S_0$ as the simulations may not reach the very fine levels at which the complexity analysis based on the asymptotic value of $\beta$ is relevant.

\begin{figure}[h!]
\centering
\caption{\label{fig_Barrier_B} Pathwise sensitivities (Vega), barrier call : $\mathbb V (\hat Y_l)(l) $}
\subfloat[$B=85, S_0=100$]{\label{fig_Barrier_B:B85}\includegraphics[width=0.45\textwidth]{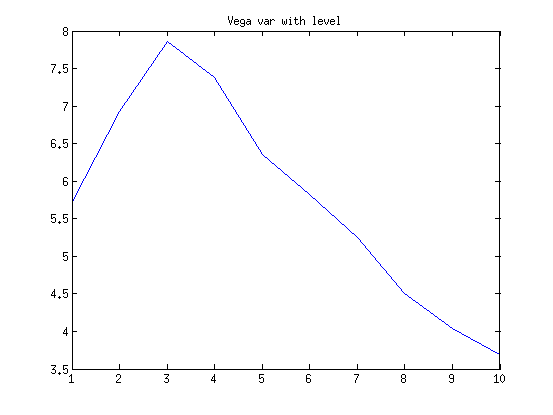}}                
\subfloat[$B=95, S_0=100$]{\label{fig_Barrier_B:B95}\includegraphics[width=0.45\textwidth]{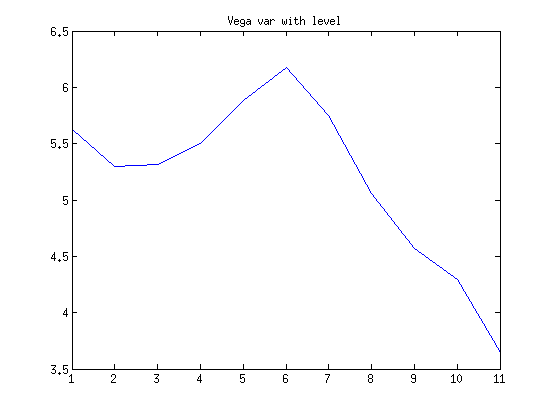}}
\end{figure}

In the case represented in figure \ref{fig_Barrier_B:B95} the variance first increases before eventually converging towards 0. This illustrates the fact that in some cases it may be interesting to start the multilevel algorithm at a level $l>0$ where problems related specifically to coarseness do not appear and where the variance is low enough for our application (it must be at least lower than than the variance of the equivalent monolevel estimator).

The variance's bad behaviour is related to the distribution of paths leaking out of the barrier over time: we plot in figure \ref{fig_Barrier_cross} the density of first barrier crossings on the time interval $[0,T]$.
\begin{figure}[h!]
\centering
\caption{\label{fig_Barrier_cross} Density of first barrier crossings on [0,1]}
\subfloat[$B=85, S_0=100$]{\label{fig_Barrier_cross:B85}\includegraphics[width=0.45\textwidth]{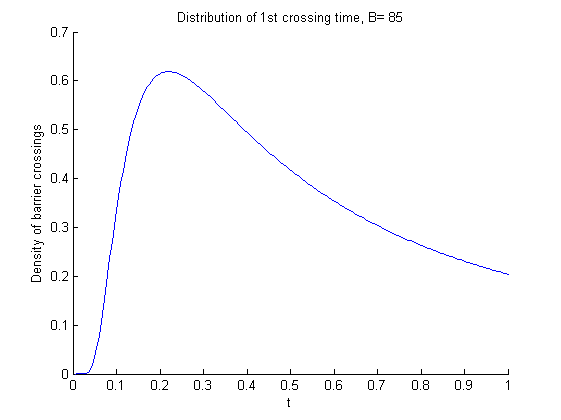}}                
\subfloat[$B=95, S_0=100$]{\label{fig_Barrier_cross:B95}\includegraphics[width=0.45\textwidth]{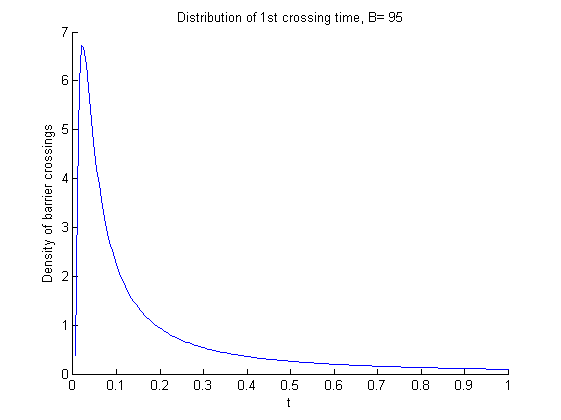}}
\end{figure}
We can show analytically that the width of the observed ``peak'' of the crossing density function is 
\begin{equation}
\tau=O\left(\frac{\log(S_0/B)}{\sigma} \right)^2 
\end{equation} 
This means that as $B$ tends to $S_0$, almost all paths going ``down and out'' do so in a very short time interval $[0,\tau]$. On timesteps $[t_n,t_{n+1}]$ outside of this interval, most paths are far away from the barrier and both $\frac{\partial \hat p_n}{\partial \theta}$ and $\frac{\partial p_n}{\partial \theta}$ are close to $0$ and hardly contribute to the variance of the multilevel estimator.

Morally the problem at the low levels is that the timesteps are much too large compared to the characteristic time $\tau$: the interval that is responsible for most of $\mathbb V(\hat Y_l)$ is covered by only one step as long as $h_l \ge \tau$. 

We hope to address this issue with adapted timestepping. Instead of taking constant timesteps of width $h=T/N$ we use power timesteps to refine the discretisation in the time interval $[0,\tau]$. We write $t=u^\gamma$ and then split $[0,T^{(1/\gamma)}] $ into equal steps of $u$. We want more steps in  $[0,\tau]$. Taking half of all timesteps in $[0,\tau] $ means that we must choose $\gamma$ such that $u=1/2$ corresponds to $\tau$, that is
\begin{equation}
\left(\frac{1}{2}\right)^{\gamma}=\left(\frac{\log(S_0/B)}{\sigma}\right)^2
\Rightarrow \gamma=\frac{2}{\log 2} \left[\log \sigma - \log \left(\log (S_0/B) \right) \right]
\end{equation}
%When $B$ is sufficiently close to $S_0$, $\gamma>1$: as expected this makes the first timesteps finer.

Figure \ref{fig_Barrier_powerstep} shows for $B=95$ the evolution of $\mathbb V(\hat Y_l)$ for  $\delta$ (fig. \ref{fig_Barrier_powerstep:delta}) and for $\nu$ (fig. \ref{fig_Barrier_powerstep:vega}) with constant timesteps (red) and with power timesteps (blue). We see that with power timesteps, the variance is lower at fine levels and reaches its asymptotic convergence speed faster than before. Nevertheless we note that these benefits may be practically cancelled by the higher variance of the method at the coarsest levels, especially if we need rough estimates and stay at very low levels.

\begin{figure}[H]
\centering
\caption{\label{fig_Barrier_powerstep} $\mathbb V(\hat Y_l)$, $B=95$, power (blue) and constant (red) steps}
\subfloat[$\delta$]{\label{fig_Barrier_powerstep:delta}\includegraphics[width=0.5 \textwidth]{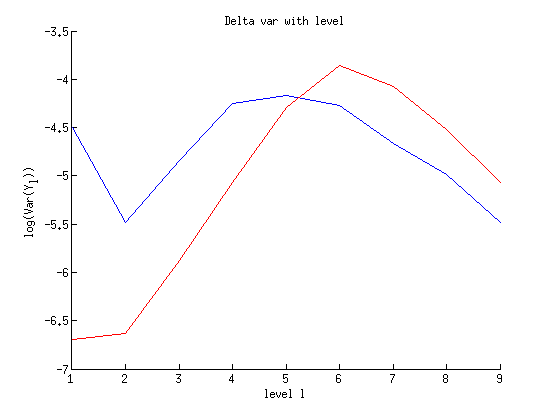}}
\subfloat[$\nu$]{\label{fig_Barrier_powerstep:vega}\includegraphics[width=0.5\textwidth]{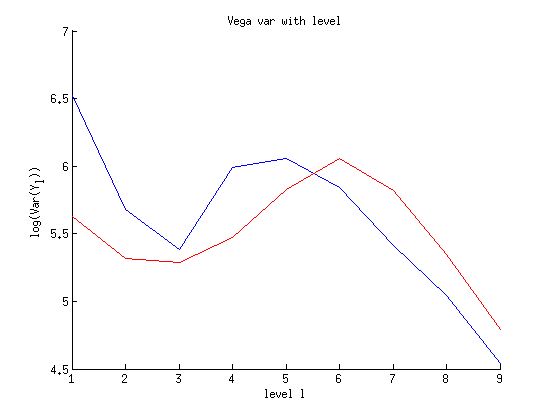}}
\end{figure}

\section*{Conclusion and future work}
In this paper we have shown for a range of cases how multilevel techniques could be used to reduce the computational complexity of Monte Carlo Greeks.

Smoothing a Lipschitz payoff with conditional expectations reduces the complexity down to $O(\epsilon^{-2})$. From this technique we derive the Path splitting and Vibrato methods: they offer the same efficiency and avoid intricate integral computations. Payoff smoothing and Vibrato  also enable us to extend the computation of Greeks to discontinuous payoffs where the pathwise sensitivity approach is not applicable. Numerical evidence shows that with well-constructed estimators these techniques provide computational savings even with exotic payoffs.

So far we have mostly relied on numerical estimates of $\beta$ to estimate the complexity of the algorithms. Our current analysis is somewhat crude ; this is why our current research now focuses on getting a rigorous numerical analysis of the algorithms' complexity.

%\section*{Figures}
%\label{burgos:6}
%
%\begin{figure}[h!]
%\sidecaption
%\includegraphics[scale=.65]{PwS01.png}
%\caption{Pathwise sensitivities}
%\label{burgos:fig1}       % Give a unique label
%\end{figure}

% sample references
%%%%%%%%%%%%%%%%%%%%%%%% Springer-Verlag %%%%%%%%%%%%%%%%%%%%%%%%%%
%
% BibTeX users please use
% \bibliographystyle{}
% \bibliography{}

\begin{thebibliography}{99.}%

\bibitem[1]{burgos:ag07}
A.~Asmussen and P.~Glynn.
\newblock {\em Stochastic Simulation}.
\newblock Springer, New York, 2007.

\bibitem[2]{burgos:giles07b}
M.B. Giles.
\newblock Improved multilevel {M}onte {C}arlo convergence using the {M}ilstein
  scheme.
\newblock In A.~Keller, S.~Heinrich, and H.~Niederreiter, editors, {\em Monte
  Carlo and Quasi-Monte Carlo Methods 2006}, pages 343--358. Springer-Verlag,
  2007.

\bibitem[3]{burgos:giles08b}
M.B. Giles.
\newblock Multilevel {M}onte {C}arlo path simulation.
\newblock {\em Operations Research}, 56(3), pages 607--617, 2008.

\bibitem[4]{burgos:giles09}
M.B. Giles.
\newblock Vibrato {M}onte {C}arlo sensitivities.
\newblock In P.~L'Ecuyer and A.~Owen, editors, {\em Monte Carlo and Quasi-Monte
  Carlo Methods 2008}, pages 369--382. Springer, 2009.

\bibitem[5]{burgos:glasserman04}
P.~Glasserman.
\newblock {\em {M}onte {C}arlo Methods in Financial Engineering}.
\newblock Springer, New York, 2004.

\bibitem[6]{burgos:kp92}
P.E. Kloeden and E.~Platen.
\newblock {\em Numerical Solution of Stochastic Differential Equations}.
\newblock Springer, Berlin, 1992.

%\bibitem[6]{burgos:Protter}
%P.~Protter.
%\newblock {\em Stochastic integration and differential equations}.
%\newblock Springer, 1990.

%\bibitem[7]{burgos:tt90}
%D.~Talay and L.~Tubaro.
%\newblock Expansion of the global error for numerical schemes solving
%  stochastic differential equations.
%\newblock {\em Stochastic Analysis and Applications}, 8:483--509, 1990.
  
%\bibitem[8]{burgos:giles09b}
%M.B. Giles.
%\newblock {M}ultilevel {M}onte {C}arlo for Basket Options.
%\newblock In M.D.Rossetti, R.R. Hill, B.Johansson, A.Dunkin, and R.G. Ingalls, editors, {\em Proceedings of the 2009 Winter Simulation Conference}, 2009.



%\bibitem[10]{burgos:bt95}
%V.~Bally and D.~Talay.
%\newblock The law of the {E}uler scheme for stochastic differential equations, {I}: convergence rate of the distribution function.
%\newblock {\em Probability Theory and Related Fields}, 104(1):43--60, 1995.

\end{thebibliography}

%\biblstarthook{}

\end{document}